\documentclass[letterpaper,twocolumn, conference]{IEEEtran}

\usepackage{fullpage,graphicx,url,setspace}

\usepackage{xparse}
\usepackage{stmaryrd}
\usepackage{cite}
\usepackage{bm}
\usepackage{color}
\usepackage[small,bf]{caption}
\usepackage{amsmath,verbatim,amssymb,amsfonts,amscd, graphicx, algorithmic, algorithm}
\usepackage{amsmath,amsthm}
\usepackage{mathtools}
\mathtoolsset{mathic}
\usepackage{microtype}
\usepackage{wrapfig}

\usepackage{ifxetex,ifluatex}
\ifxetex
\usepackage{xltxtra}
\else
\ifluatex
\usepackage{luatextra}
\else
\usepackage{amssymb}
\usepackage{algorithmic, algorithm}
\usepackage[utf8]{inputenc}

\csname else\expandafter\endcsname
\romannumeral -`0
\fi
\fi
\iftrue\relax%
\defaultfontfeatures{Ligatures=TeX}
\usepackage[math-style=TeX, vargreek-shape=TeX]{unicode-math}
\fi

\NewDocumentCommand{\entropy}{om}{\mathbb{H}\left[#2
    \IfValueT{#1}{\,\middle|\,#1}\right]}
\NewDocumentCommand{\bentropy}{lm}
  {\widetilde{\mathbb{H}}#1\left[#2\right]}

\NewDocumentCommand{\mutualInfo}{omm}{\mathbb{I}\left[#2;#3
    \IfValueT{#1}{\,\middle|\,#1}\right]}

\usepackage{tikz}
\usetikzlibrary{shapes.geometric, positioning, shadows}
\usepackage{graphicx,color}

\newtheorem{theorem}{Theorem}

\newtheorem{lemma}{Lemma}

\newtheorem{remark}{Remark}

\newtheorem{corollary}{Corollary}
\usepackage{enumitem}
\newlist{enumerate*}{enumerate*}{1}
\setlist[enumerate*]{label=(\arabic*)}

\newcommand{\ben}{\begin{eqnarray}}

\newcommand{\een}{\end{eqnarray}}

\newcommand{\yc}[1]{{\color{black}{#1}}}

\DeclareMathOperator*{\argmin}{arg\,min}

\def\U{{\mathcal{U}}}
\def\Z{{\mathcal{Z}}}
\def\H{{\mathcal{H}}}
\newcommand{\half}{ \mbox{\small$\frac{1}{2}$}}
\def\Det{{\mbox{Det}}}
\def\Tr{{\mbox{Tr}}}

\title{Robust Sequential Change-Point Detection by \\Convex Optimization}

\author{\IEEEauthorblockN{Yang Cao and Yao Xie}
\IEEEauthorblockA{H. Milton Stewart School of Industrial and Systems Engineering\\
Georgia Institute of Technology\\
\{caoyang, yao.xie\}@gatech.edu }}

\begin{document}
\maketitle

\begin{abstract}

We address the computational challenge of finding the robust sequential change-point detection procedures when the pre- and post-change distributions are not completely specified. Earlier works \cite{veeravalli1994minimax, unnikrishnan2011minimax} establish the general conditions for robust procedures which include finding a pair of least favorable distributions (LFDs). However, in the multi-dimensional setting, it is hard to find such LFDs computationally. We present a method based on convex optimization that addresses this issue when the distributions are Gaussian with unknown parameters from pre-specified uncertainty sets. We also establish theoretical properties of our robust procedures, and numerical examples demonstrate their good performance\footnote{Proofs to theorems can be found in the arXiv version of this paper: arXiv:1701.06952.}. 

\end{abstract}

%

\section{Introduction}

Sequential detection of an abrupt change has wide applications such as statistical quality control and network security monitoring. In the classic settings, one obtains a sequence of observations of a signal of which the distribution changes at some unknown point in time, referred to as the ``change-point''. The goal is to detect the change as quickly as possible, subject to the false-alarm constraint. With the ever growing complexity of systems and enlarging number of sensors to monitor the systems, multi-sensor change-point detection has become a quite important subject (see, \cite{tartakovsky2002efficient}, \cite{mei2010efficient} and \cite{xie2013sequential}).

Classic sequential change-point detection assumes that the distributions before and/or after the change-point are completely specified (e.g., the classic CUSUM for one-sensor \cite{page1954continuous,lorden1971procedures} and Shiryaev-Roberts procedure \cite{shiryaev1963optimum,lorden1971procedures}). Under this setting, CUSUM is optimal (see, e.g., \cite{moustakides1986optimal}). However, CUSUM procedure is known to be sensitive to the misspecified distributions \cite{stoumbos2000state}. 

Robust detector dates back to Huber's seminal work \cite{huber1965robust}. Subsequent follow-up work considers robust detector such as \cite{Moustakides85}. Huber considers the class of all symmetric densities that satisfies the so-called $\epsilon$-contamination model, with symmetric but unknown contaminations on the nominal distributions.  
Robust sequential change detection based on the above framework was considered in \cite{crow1994robust}. The more recent contributions \cite{veeravalli1994minimax, unnikrishnan2011minimax} introduce a so-called Joint Stochastic Boundedness (JSB), under which one can identify a pair of least favorable distributions (LFDs) from the uncertainty classes such that the CUSUM procedure designed for the LFDs is optimal for the robust problem in the minimax sense. 
However, in the multi-dimensional setting, there remains the computational challenge to establish robust sequential detection procedures or to find the LFDs. Closed-form LFDs can only be found for a few special one-dimensional cases (e.g,\cite{huberrobust} and \cite{levy2008principles}).  Moreover, the JSB condition in \cite{unnikrishnan2011minimax} is defined on the real line;  direct extension of JSB to multi-dimensional setting becomes quite restrictive even in very simple cases, illustrated in the following example. Consider two bivariate normal distributions. Assume that $\Sigma$ is a positive-definite matrix in $\mathbb{R}^{2\times 2}$, and we would like to detect a possible transition from 
$
\mathcal{P}_0 = \{\mathcal{N}(0,\Sigma)  \},
$
to a family of distributions
$
\mathcal{P}_1 = \{\mathbb{P}\mid \mathbb{P} = \mathcal{N}(\mu_1, \Sigma), \|\mu_1-(10,10)^T\|_2 \leq 1, \mu_1 \in \mathbb{R}^2 \}.
$
In this case, it is impossible to find a distribution in $\mathcal{P}_1$ that is stochastically larger than any other distribution in $\mathcal{P}_1$ due the following Lemma \ref{stochastical_order} which satisfies the JSB condition (also see Fig. \ref{fig:counterexample} for the illustration). 
\begin{lemma}[Theorem 5 in \cite{muller2001stochastic} ]
Let $X \sim \mathcal{N}(\mu, \Sigma)$ and $X' \sim \mathcal{N}(\mu', \Sigma')$ be $n$-dimensional normally distributed random vectors. Then $X'$ is stochastically larger than $X$ if and only if $\mu'^{(i)} \geq \mu^{(i)}, \mbox{ for all } 1\leq i \leq n$ and $\Sigma = \Sigma'$, where $\mu^{(i)}$ denotes the $i$th entry of $\mu$. 
\label{stochastical_order}
\end{lemma}

\begin{figure}[h!]
\vspace{-0.2in}
\begin{center}
\begin{tabular} {cc}
\includegraphics[width = .5\linewidth]{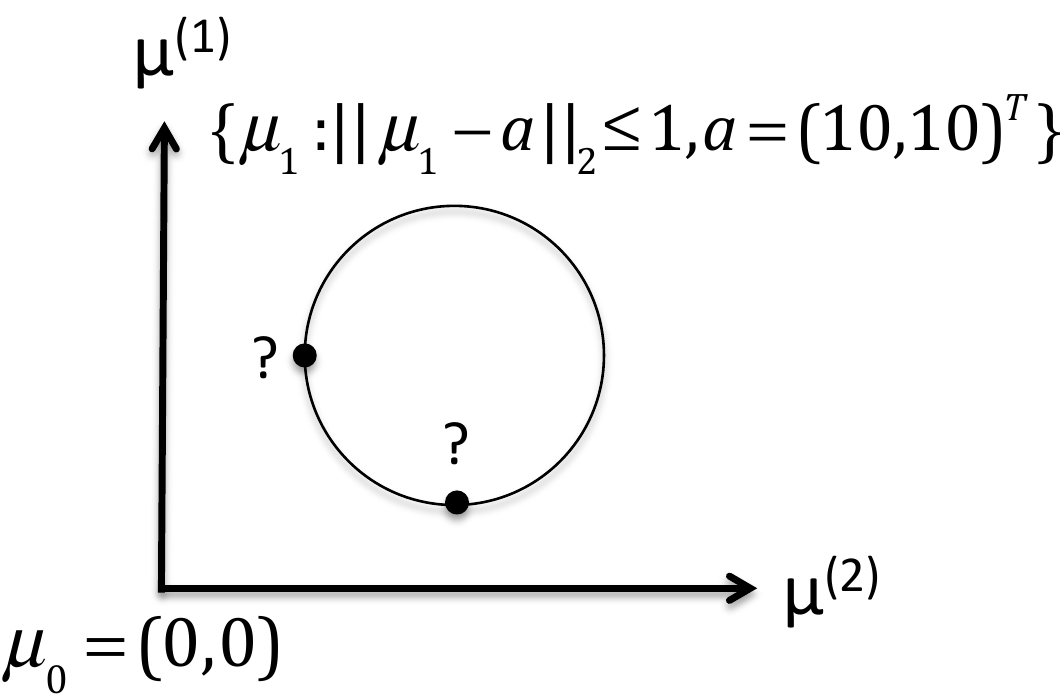}
\end{tabular}
\end{center}
\vspace{-0.1in}
\caption{It is impossible to find a point on the circle of which all the entries are larger than those of other points on the circle.}
\label{fig:counterexample}
\vspace{-0.1in}
\end{figure}

In this paper, we present a method of establishing the robust procedure by solving a convex optimization program. Given convex sets for parameters under the null and the alternative distributions, instead of identifying LFDs, we find a pair of {\it least favorable} parameters such that the Hellinger distance between the corresponding distributions from uncertainty sets are minimized. In this paper, we consider detecting the change in the mean vector and covariance matrix of a multivariate normal distribution, and hence, we may restrict our attention to linear and quadratic ``detectors'' (the methodology can be generalized to other sub-Gaussian distributions, see, e.g., \cite{guigues2016change}). Then a CUSUM procedure is defined for the pair of parameters solved from the optimization problem. We analyze the theoretical properties of our procedure. Note that since we use parametric models and represent uncertainty as ``uncertainty sets'' for the parameters, this is different from the previous work that identifies LFDs where the uncertainty class is represented by a set of probability functions. A benefit of our approach is that it leads to computationally more efficient methods.

Our approach is motivated by the recent work using convex optimization for hypothesis testing \cite{goldenshluger2015hypothesis,guigues2016change}. The difference of these approaches from our work is that they consider finite time-horizon sequential change-point detection problem and treat it as multiple hypothesis tests. Since for each time $t$, there are $k$ possible change-point locations, for a fixed time horizon $t\leq T$ there are a finite number of hypotheses. Then one may design a test such that the probability of error for each of the hypothesis is uniformly controlled and the total probability of error is less than a given level $\epsilon$. This approach is not convenient for infinite horizon setting considered in the usual sequential change-point detection problem. In this paper, we essentially proposal another approach for using the framework in \cite{goldenshluger2015hypothesis,guigues2016change} for infinitely horizon sequential hypothesis testing, where one can determine the Average Run Length (ARL) and Expected Detection Delay (EDD) analytically.

\vspace{-0.05in}
\section{Formulation}

\vspace{-0.05in}
\subsection{General setup}

Assume that we observe a sequence of observations $\{\xi_i\}_{i=1}^{\infty}$ that take values in $\mathcal{X}$. Denote $\mathcal{P}(\mathcal{X})$ as the set of all the probability distributions on $\mathcal{X}$ and assume that there are two known distributions $\nu_0, \nu_1 \in \mathcal{P}(\mathcal{X})$. If there is no change, the observations are drawn i.i.d. from distribution $\nu_0$. The probability and expectation in this case are denoted by $\mathbb{P}_{\infty}^{\nu_0}$ and $\mathbb{E}_{\infty}^{\nu_0}$, respectively. Alternatively, the i.i.d. observations $\xi_i \sim \nu_0$ for $i=1,\ldots, \kappa-1$, and at some \emph{unknown} change-point $\kappa$, the distributions of the observations switch abruptly to $\nu_1$, namely, $\xi_i \sim \nu_1$ for $i=\kappa, \kappa+1,\ldots$. The observations are  independent conditioned on the change-point $\kappa$. The probability and expectation in this case are denoted by $\mathbb{P}_{\kappa}^{\nu_0,\nu_1}$ and $\mathbb{E}_{\kappa}^{\nu_0, \nu_1}$, respectively.  In particular, $\kappa$ = 0 denotes an immediate change occurring at the initial time.

A sequential change detection procedure is characterized by a stopping time $T$ with respect to the observation sequence. To evaluate the performance of the detection procedure $T$, two performance measures are widely used: the average run length (ARL) and the expected detection delay (EDD). There are three commonly used mathematical formulations about ARL and EDD: Lorden's worst-case formulation in \cite{lorden1971procedures}, Pollak's average worst-case formulation in \cite{pollak1985optimal} and the Bayesian formulation in \cite{shiryaev1963optimum}. In this paper, we adopt the Lorden's formulation, where the worst-case EDD of a detection procedure $T$ is defined as follows:
\begin{equation}
\mbox{WDD}(T) = \sup_{k \geq 1} \mbox{esssup}~ \mathbb{E}_{k}^{\nu_0, \nu_1}\left[ (T-k+1)^+ \mid \mathcal{F}_{k-1} \right],
\label{WDD}
\end{equation}
where $(x)^+ = \max(x,0)$. The quantity in (\ref{WDD}) is called the worst-case EDD as a result of the two supreme appearing in (\ref{WDD}). The first supreme means that the detection delay is taken over all possible locations of the change-point $k$ and the second essential supreme means that the detection delay is taken over all possible realizations of the observations before the change-point $k$. ARL can be interpreted as the mean time between two false alarms, denoted by $\mathbb{E}_{\infty}^{\nu_0}[T]$. In practice, one usually fixes a lower bound $\gamma$ for the ARL and denotes $C(\gamma)$ as the set of stopping times with ARL larger than $\gamma>0$, in other words, $C(\gamma) = \left\{T: \mathbb{E}_{\infty}^{\nu_0}[T]\geq \gamma \right\}$. Then, our goal is to solve the following problem:
\begin{equation}
\min_{T \in C(\gamma)} ~\mbox{WDD}(T). 
\label{opti1}
\end{equation}
In \cite{lorden1971procedures} and \cite{moustakides1986optimal}, it has been proven that the cumulative sum (CUSUM) procedure \cite{page1954continuous} is both the asymptotically optimal solution as $\gamma \rightarrow \infty$ and the exact optimal solution to (\ref{opti1}) for any given $\gamma>0$. Hence, in the following, we will focus on CUSUM-type procedures.

\yc{Now we consider the case when $\nu_0$ and $\nu_1$ are not specified exactly but belong to two classes of distributions $\mathcal{P}_0, \mathcal{P}_1 \in \mathcal{P}(\mathcal{X})$, respectively (such definitions have been considered in \cite{unnikrishnan2011minimax}).}
Denote $C(\mathcal{P}_0, \gamma) = \{ T: \mathbb{E}_{\infty}^{\nu_0}[T] \geq \gamma, \forall \nu_0 \in \mathcal{P}_0 \}$ as the set of all candidate stopping times whose ARL is lower bounded by $\gamma$. Then our goal is to solve the following robust version of (\ref{opti1}):
\begin{equation}
\min_{T \in C(\mathcal{P}_0, \gamma)}~ \sup_{\nu_0 \in \mathcal{P}_0, \nu_1 \in \mathcal{P}_1} \mbox{WDD}(T).
\label{opti4}
\end{equation}

In the following, we specify the uncertainty set for parametric distributions, by assuming {\it convex uncertainty sets} for the parameters. This is a versatile formulation which finds useful in many practical situations (see, e.g., \cite{Robust09}).




{\it Mean  change:} Assume that we observe a sequence of $d$-dimensional multivariate normal distribution with a known covariance matrix that does change. At some time $\kappa$, the mean vector switches from $\mu_0, \mu_0 \in \mathcal{M}_0$ to $\mu_1, \mu_1 \in \mathcal{M}_1$, where $\mathcal{M}_0$ and $\mathcal{M}_1$ are two known convex sets in $\mathbb{R}^d$ that are user-specified beforehand. The observations are  independent conditioned on the change-point $\kappa$. Mathematically, we formulate the problem as the following hypothesis testing problem:
\begin{equation}
\begin{split}
H_0  : \quad &\xi_i \sim \mathcal{N}(\mu_0, \Sigma), \mu_0 \in \mathcal{M}_0, i=1,2,\ldots  \\
H_1  : \quad &\xi_i \sim \mathcal{N}(\mu_0, \Sigma), \mu_0 \in \mathcal{M}_0, i=1,2,\ldots,\kappa, \\
&\xi_i \sim \mathcal{N}(\mu_1,\Sigma), \mu_1\in \mathcal{M}_1, i=\kappa+1,\kappa+2,\ldots,
\end{split}
\label{testing}
\end{equation}
where $\Sigma$ is the known positive definite covariance matrix. 
%
%
Here, the mean vector $\mu_0$ and $\mu_1$ can be any element in the convex sets $\mathcal{M}_0$ and $\mathcal{M}_1$, respectively. For example, in the context of quality control, $\mathcal{M}_0$ can be defined as the set of all the allowable mean vectors if the system is in-control and $\mathcal{M}_1$ denotes the set of all the possible mean vectors if the system is out-of-control. Our goal is to identify the occurrence of the change as fast as possible subject to the false alarm constraints.

{\it Covariance matrix change:} Similarly, we may come up with a formulation when both the mean and the covariance matrix of the observations change. Assume a sequence of $d$-dimensional multivariate normal observations. At some time $\kappa$, the mean vector changes from $\mu_0, \mu_0 \in \mathcal{M}_0$ to $\mu_1, \mu_1 \in \mathcal{M}_1$ and the covariance matrix changes from $\Theta_0, \Theta_0 \in \mathcal{U}_0$ to $\Theta_1, \Theta_1 \in \mathcal{U}_1$, where $\mathcal{M}_0$ and $\mathcal{M}_1$ are two known convex sets in $\mathbb{R}^d$, $\mathcal{U}_0$ and $\mathcal{U}_1$ are two known convex sets in $\mathbb{S}_+^d$, which are user-specified beforehand. We may formulate this problem as a hypothesis test similar to above. 



Even if the formulation for the covariance case looks similar to the formulation (\ref{testing}), here the problem is much more difficult than (\ref{testing}). For instance, a natural approach is to use sample mean and sample covariance matrices from the in-control and out-of-control data (there usually are these training data available in certain form) as the parameters before and after the change when designing the procedures. Then the uncertainty sets represents the estimation ``precision'', which depend on the sample size and how the estimators are constructed. Mean vectors can usually be estimated up to good precision. However, it is much harder to estimate high-dimensional covariance matrix accurately (see, e.g, \cite{bickel2008regularized}, \cite{ravikumar2011high}, 
and \cite{fan2013large}). 
Fortunately,  most of the existing methods can guarantee that the true covariance matrix belongs to a convex set in $\mathbb{S}_+^d$, which enables us to reasonably construct uncertainty sets for covariance matrices.  


\section{Main results}


\subsection{Robust procedure for detecting mean change}

For the robust version for mean shift detection (\ref{testing}), we consider a CUSUM-type procedure. CUSUM procedure needs specified likelihood ratio for two {\it singleton} pre-change and post-change distributions. Here, we solve a convex optimization problem to identify an appropriate pairs of parameters for the pre-change and post-change distributions, and use them to form the CUSUM procedure. 

Let $\mathcal{P}_0 = \{ \mathcal{N}(\mu_0, \Sigma), \mu \in \mathcal{M}_0  \}$ and $\mathcal{P}_1 = \{ \mathcal{N}(\mu_1, \Sigma), \mu \in \mathcal{M}_1  \}$. Specifically, denote $(\mu_0^*,\mu_1^*)$ as the solution to the following convex optimization problem:
\begin{equation}
(\mu_0^*,\mu_1^*) = \argmin_{\mu_0 \in \mathcal{M}_0, \mu_1 \in \mathcal{M}_1} (\mu_0-\mu_1)^T \Sigma^{-1} (\mu_0-\mu_1).
\label{best_mean}
\end{equation}
In other words, $\mu_0^*$ and $\mu_1^*$ are two points in $\mathcal{M}_0$ and $\mathcal{M}_1$ with the minimal Mahalanobis distance. 

Our detection procedure is given as follows:
\begin{equation}
T_1 = \inf \left\{ t>0: \max_{1\leq k \leq t} \sum_{i=k}^t \frac{1}{2}L^*(\xi_i) \geq b \right\},
\label{CUSUM-like}
\end{equation}
where $L^*$ denotes the likelihood ratio between $\nu_1^* \sim \mathcal{N}(\mu_1^*, \Sigma)$ and $\nu_0^* \sim \mathcal{N}(\mu_0^*, \Sigma)$. The threshold $b$ is chosen such that $\mathbb{E}_{\infty}^{\nu_0}[T_1] \geq \gamma$ for all $\nu_0 \in \mathcal{P}_0$ and a prescribed lower bound $\gamma$ for ARL. We can show the following relationship between $\gamma$ and $b$, which offers a guideline about how to determine $b$ given any $\gamma$.
\begin{theorem}[ARL]
For any $\nu_0 \in \mathcal{P}_0$, for the detection procedure $T_1$ defined in (\ref{CUSUM-like}), we have that 
$\mathbb{E}_{\infty}^{\nu_0}[T_1] \geq \gamma$ as long as 
\begin{equation}
b \geq \log \gamma + \log \frac{\epsilon^*}{1-\epsilon^*},
\label{b_eqn}
\end{equation}
where 
\begin{equation}
\epsilon^* = \exp(-\frac{1}{8} (\mu_0^* - \mu_1^*)^T \Sigma^{-1} (\mu_0^* - \mu_1^*) ).
\label{error}
\end{equation} 
\label{ARL1}
\end{theorem}
%
\begin{remark}
When $\mathcal{P}_0 = \{ \nu_0 \}$ and $\mathcal{P}_1 = \{ \nu_1 \}$ are two singletons, $T_1$ is just the classic CUSUM procedure and the classic analysis tells us that if $b\geq \log \gamma$ then $\mathbb{E}_{\infty}^{\nu_0}[T_1] \geq \gamma$. The additional second term $\log (\epsilon^*/(1-\epsilon^*))$ in (\ref{b_eqn}) can be seen as a cost for the uncertainty. Specifically, $\epsilon^*$ is the upper bound for the Type-I and Type-II error for the one sample composite hypothesis testing problem: $H_0: \xi \sim \nu_0, \nu_0 \in \mathcal{P}_0$ versus $H_1: \xi \sim \nu_1, \nu_1 \in \mathcal{P}_1$.
\end{remark}

Next, we prove an upper bound for the worst-case detection delay as the threshold $b$ goes to infinity. \yc{In the following, let $o(1)$ be a vanishing term as $\gamma \rightarrow \infty$.} 

\begin{theorem}[EDD]
For any $\nu_0 \in \mathcal{P}_0$ and $\nu_1 \in \mathcal{P}_1$, for the detection procedure $T_1$ defined in (\ref{CUSUM-like}), as $b \rightarrow \infty$, we have that 
\[
\mbox{WDD}(T_1) \leq \frac{b}{1-\epsilon^*}(1+o(1)),
\]
where $\epsilon^*$ is defined in (\ref{error}) \yc{and $o(1)$ is a vanishing term as $b\rightarrow \infty$.} 
\yc{Therefore, as $\gamma \rightarrow \infty$, we can have both $\mathbb{E}_{\infty}^{\nu_0}[T_1] \geq \gamma$ and
\[
\mbox{WDD}(T_1) \leq \frac{\log \gamma}{1-\epsilon^*}(1+o(1)),
\]
where $\epsilon^*$ is defined in (\ref{error}) 
}
\label{EDD1}
\end{theorem}

\begin{remark}
Note that $1-\epsilon^*$ is just the Hellinger distance between the two multivariate normal distributions found by solving the convex optimization problem: $\mathcal{N}(\mu_0^*, \Sigma)$ and $\mathcal{N}(\mu_1^*, \Sigma)$. When $\mathcal{P}_0 = \{ \nu_0 \}$ and $\mathcal{P}_1 = \{ \nu_1 \}$ are two singletons, the classic analysis tells that the $\mbox{WDD}(T_1)$ is asymptotically upper bounded by $2b/I$, where $I$ is the Kullback-Leibler(KL) divergence between pre-change and post-change distributions. The Hellinger distance plays a similar role with the KL divergence as the denominator in Lorden's work \cite{lorden1971procedures}. Since KL divergence is known to be bounded below by Hellinger distance, our upper bound is a little bit looser. This can also be seen as the cost for uncertainty. 
\end{remark}

\begin{remark}
Define that $\bar{\nu}_0$ and $\bar{\nu}_1$ are true pre-change and post-change distributions. Since we can interpret the robust detection procedure $T_1$ as a repeated one-sided sequential probability ratio test (SPRT) between $\nu_0^* = \mathcal{N}(\mu_0^*, \Sigma)$ and $\nu_1^* = \mathcal{N}(\mu_1^*, \Sigma)$, we in fact can obtain that the WDD of $T_1$ is asymptotically upper bounded by $2b/(KL(\bar{\nu}_1 \| \nu_0^*) - KL(\bar{\nu}_1 \| \nu_1^*))$. As stated in the seminal work \cite{unnikrishnan2011minimax}, compared with the optimal CUSUM procedure between $\bar{\nu}_0$ and $\bar{\nu}_1$, $\mbox{WDD}(T_1)$ is asymptotically larger by a factor no more than 
$
KL(\bar{\nu}_0 \| \bar{\nu}_1)/(KL(\bar{\nu}_1 \| \nu_0^*) - KL(\bar{\nu}_1 \| \nu_1^*)). 
$
Furthermore, as a consequence of theorem \ref{EDD1}, for any two true pre-change and post-change distributions $\bar{\nu}_0$ and $\bar{\nu}_1$, we have that $\mbox{WDD}(T_1)$ is asymptotically larger by a factor no more than
$KL(\bar{\nu}_0 \| \bar{\nu}_1)/[2(1-\epsilon^*)].$
When the Mahalanobis distance between $\mathcal{M}_0$ and $\mathcal{M}_1$ increases, $\epsilon^*$ in (\ref{error}) becomes smaller and then factor above decreases, which means that our procedure moves closer to the optimal one. This is consistent with our intuition that one can detect the change more easily when the change is more obvious.

\end{remark}

\subsection{Robust procedure for detecting covariance change}

Next, consider the case when both the mean vector and the covariance matrix of a multivariate normal distribution change and they belong to some uncertainty sets. In this case, we may consider linear and quadratic detectors, parameterized by vector $h$ and matrix $H$ defined below, as suggested in \cite{guigues2016change}. We include the original derivation from \cite{guigues2016change} below. 

First we define the cost function, which can be viewed as exponential loss function which relates to the type-I and type-II error in the test (in the fixed sample size scenario). 
Let $\| \cdot \|$ denote the spectral norm and $\| \cdot\|_F$ the Frobenius norm, respectively. Let $\U$ be a convex compact set contained in the interior of the cone $S^d_+$ of positive semidefinite $d \times d$ matrices in the space $S^d$ of symmetric $d\times d$ matrices. Let $\Theta_*\in S^{d}_+$ be such that $\Theta_*\succeq \Theta$ for all $\Theta\in\U$, and let $\delta\in[0,2]$ be such that
\begin{equation}\label{56delta}
\|\Theta^{1/2}\Theta_*^{-1/2}-I_d\|\leq\delta,\;\;\forall \Theta\in\U.
\end{equation}
Let $\mathcal{Z}$ be a nonempty convex compact subset of the set $\mathcal{Z}^+=\{Z\in S^{d+1}_+:Z_{d+1,d+1}=1\}$, and let
\begin{equation}\label{phiZ}
\phi_\mathcal{Z}(Y) \triangleq \max_{Z\in\mathcal{Z}} \mbox{Tr}(ZY)
\end{equation}
be the support function of $\mathcal{Z}$; this function is used in the following definition of $\Phi_\mathcal{Z}$. 
These specify the closed convex set
\begin{equation}\label{PcH}
\begin{split}
\H&=\H^\beta\\
&:=\{(h,H)\in \mathbb{R}^d \times \mathbb{S}^d:-\beta\Theta_*^{-1}\preceq H\preceq \beta \Theta_*^{-1}\},
\end{split}
\end{equation}
and the function $\Phi_{\Z}:\,\H\times\U\to\mathbb{R}$,
\begin{equation}\label{phi}
\begin{split}
&\Phi_{\Z}(h,H;\Theta)=\\
&-\frac{1}{2}\log\Det(I-\Theta_*^{1/2}H\Theta_*^{1/2})+\frac{1}{2} \Tr([\Theta-\Theta_*]H)\\
+&{\delta(2+\delta)\over 2(1-\|\Theta_*^{1/2}H\Theta_*^{1/2}\|)}\|\Theta_*^{1/2}H\Theta_*^{1/2}\|_F^2\\
+&{1\over 2}\phi_\Z\left(\left[\hbox{\small$\left[\begin{array}{c|c}H&h\cr\hline h^T&\end{array}\right]+
\left[H,h\right]^T[\Theta_*^{-1}-H]^{-1}\left[H,h\right]$}\right]\right).
\end{split}
\end{equation}
Then, we have that $\Phi_{\Z}$ is continuous on its domain, convex in $(h,H)\in\H$ and concave in $\Theta\in\U$. 

Next, we specify the uncertainty sets for the pre-change and post-change multivariate normal distributions. Given two collections of data as above: $(\mathcal{U}_\chi, \Theta_*^{(\chi)}, \beta_{\chi}, \mathcal{Z}_{\chi}), \chi =0,1$, we define that
\begin{equation}\label{Gchi}
\begin{split}
\mathcal{G}_\chi = &\{N(\mu,\Theta): \Theta\in\U_\chi \ \\
& \quad \exists u: \mu=[u;1], [u;1][u;1]^T\in\Z_\chi \},\,\chi=0,1.
\end{split}
\end{equation}

Now to solve for the quadratic detector $(h,  H)$, which will be applied on each individual samples and then used to construct the CUSUM recursion, we consider the convex-concave saddle point problem
\begin{equation}\label{SPPLift}
\begin{split}
&\mathcal{SV}=\min\limits_{(h,H)\in\H_0\cap\H_1}\max\limits_{\Theta_0\in\U_0,\Theta_1\in\U_1} \\
&\qquad \qquad\underbrace{\half\left[\Phi_{\Z_0}(-h,-H;\Theta_0)+
\Phi_{\Z_1}(h,H;\Theta_1)\right]}_{\Phi(h,H;\Theta_0,\Theta_1)}.
\end{split}
\end{equation}
A saddle point $(H_*,h_*;\Theta_0^*,\Theta_1^*)$ in this problem does exist, which corresponds to the parameters of the quadratic detector and the picked worst-case parameters, which can be solved using a semi-definite program (SDP) solver. We obtain the following quadratic detector
\begin{equation}\label{iquaddet}
\begin{split}
\phi^*(\xi)=&\frac{1}{2}\xi^TH_*\xi+h_*^T\xi+\\
&\frac{1}{2}\left[\Phi_{\Z_0}(-h_*,-H_*;\Theta^*_0)
-\Phi_{\Z_1}(h_*,H_*;\Theta^*_1)\right],
\end{split}
\end{equation}
Given above (which is pre-solved before we have seen any data), now given a sequence of data, we may evaluate $\phi^*$ in (\ref{iquaddet})  for each sample and define our detection procedure as follows:
\begin{equation}
T_2 = \inf \left\{ t>0: \max_{1\leq k \leq t} \sum_{i=k}^t (-\phi^*(\xi_i)) \geq b \right\},
\label{CUSUM-like-2}
\end{equation}
where $b$ is a prescribed threshold.  

\begin{corollary}[ARL]
For any $\nu_0 \in \mathcal{G}_0$, for the detection procedure $T_2$ defined in (\ref{CUSUM-like-2}), we have that 
$\mathbb{E}_{\infty}^{\nu_0}[T_2] \geq \gamma$ as long as 
\[
b \geq \log \gamma + \log \frac{\epsilon^*}{1-\epsilon^*},
\]
where 
\begin{equation}
\epsilon^* = \exp(\mathcal{SV})
\label{error-2}
\end{equation} 
and $\mathcal{SV}$ is defined in (\ref{SPPLift}).
\label{ARL2}
\end{corollary}

\begin{corollary}[EDD]
For any $\nu_0 \in \mathcal{G}_0$ and $\nu_1 \in \mathcal{G}_1$, for the detection procedure $T_2$ defined in (\ref{CUSUM-like-2}), as $b \rightarrow \infty$, we have that 
\[
\mbox{WDD}(T_2) \leq \frac{b}{1-\epsilon^*}(1+o(1)),
\]
where $\epsilon^*$ is defined in (\ref{error-2}) \yc{and $o(1)$ is a vanishing term as $b\rightarrow \infty$.}
\yc{Therefore, as $\gamma \rightarrow \infty$, we can have both $\mathbb{E}_{\infty}^{\nu_0}[T_1] \geq \gamma$ and
\[
\mbox{WDD}(T_2) \leq \frac{\log \gamma}{1-\epsilon^*}(1+o(1)),
\]
where $\epsilon^*$ is defined in (\ref{error-2}) \yc{and $o(1)$ is a vanishing term as $\gamma \rightarrow \infty$.} 
}
\label{EDD2}
\end{corollary}





\section{Numerical examples}

In this section, we  compare our procedures numerically with the corresponding classic CUSUM procedure. In all the following experiments, we set the dimension $d=30$ and choose $b$'s such that the ARL of $T_1$ and $T_{\rm CUSUM}$ are both $5000$. The classic CUSUM procedure are formed using randomly chosen pre-change and post-change distributions from the uncertainty sets. \yc{In the following, we denote $\textbf{1}$ as an all-one vector.} The comparison of the numerical example is shown in Table I.

\begin{table}[h!]
\begin{center}
\caption{Comparison of Robust and Original CUSUM.  Results are obtained from 500 Monte Carlo trials. The standard deviation is in the bracket. }
\begin{tabular}{|r|c|c|}
\hline
& mean-shift,  $\ell_1$ set & mean-shift,  $\ell_2$ set \\\hline
$T_{1}$ & 7.6 (2.3) & 10.3 (2.9) \\\hline
$T_{\rm CUSUM}$ & 32.2 (30.1) & 32.1 (31.0) \\\hline
& cov-shift, set 1 & cov-shift, set 2 (larger) \\\hline
$T_{1}$ & 9.10 (4.21) & 2.06 (0.33) \\\hline
$T_{\rm CUSUM}$ & 8.28 (5.10) & 10.28 (9.22)\\\hline
\end{tabular}
\end{center}
\vspace{-0.2in}
\end{table}

\subsection{Mean change detection}
Assume $\mathcal{M}_0 = \{0\}$ and $\Sigma = I$ in (\ref{testing}). In the first example, set $\mathcal{M}_1 = \{x \in \mathbb{R}^d: \|x-\textbf{1}\|_1 \leq 27 \}$ in (\ref{testing}). We run $1000$ experiments and for each run we choose a mean vector $\mu$ whose entries are random from $[0.1,0.5]$, then generate the post-change observations from $\mathcal{N}(\mu, I)$. For classic CUSUM, we specify the pre-change distribution as $\mathcal{N}(0,I)$ and the post-change distribution as $\mathcal{N}(\textbf{1}, I)$. Then we obtain $1000$ simulated detection delays of $T_1$ and $T_{\rm CUSUM}$. {\bf The mean and standard deviation of detection delay of $T_1$ are $7.6$ and $2.3$, and those of $T_{\rm CUSUM}$ is $32.2$ and $30.1$, respectively}. In this case, $T_1$ performs much better than $T_{\rm CUSUM}$ since it is difficult to choose a good post-change distribution in $\mathcal{M}_1$ that is close to the true post-change distribution. 

In the second example, the only difference between the second and the first example is that we replace the norm in $\mathcal{M}_1$ from $\ell_1$ to $\ell_2$. Set $\mathcal{M}_1 = \{x \in \mathbb{R}^d: \|x-\textbf{1}\|_2^2 \leq 27 \}$ in (\ref{testing}). We run $1000$ experiments, and for each run we choose a mean vector $\mu$ whose entries are random from $[0.1,0.5]$, then generate the post-change observations from $\mathcal{N}(\mu, I)$. For classic CUSUM, we specify the pre-change distribution to be $\mathcal{N}(0,I)$, and the post-change distribution to be $\mathcal{N}(\textbf{1}, I)$. Then we obtain $1000$ simulated detection delays of $T_1$ and $T_{\rm CUSUM}$. {\bf The mean and standard deviation of detection delay of $T_1$ is $10.3$ and $2.9$, and those of $T_{\rm CUSUM}$ is $32.1$ and $31.0$, respectively}. In this case, $T_1$ again performs much better than $T_{\rm CUSUM}$.


\subsection{covariance matrix change detection}

Consider $\mathcal{M}_0 = \mathcal{M}_1 = \{0\}$ and $\mathcal{U}_0 = \{I\}$. 
In the first example, we set $\mathcal{U}_1 = \{ I+\sigma V, \sigma\in [0.5,1] \}$, where $V$ is a known matrix with diagonal entries $V_{i,i}=0,i=1,\ldots,d$ and off-diagonal entries $V_{i,j} = \exp(-(i-j)^2), i,j=1,\ldots,d, i\neq j$. We run $500$ experiments and for each run we randomly choose $\sigma \in [0.5,1]$ and then generate the post-change observations from $\mathcal{N}(0, I+\sigma V)$. For classic CUSUM, we  specify the pre-change distribution as $\mathcal{N}(0,I)$ and the post-change distribution as $\mathcal{N}(0, I+0.75V)$. Then we obtain $500$ experiments for $T_2$ and $T_{\rm CUSUM}$. {\bf The mean and standard deviation of detection delay of $T_2$ is $9.10$ and $4.21$, and those of $T_{\rm CUSUM}$ is $8.28$ and $5.10$}. In this case, there is no obvious difference between the two detection procedures, which means that $T_2$ performs almost as well as classical CUSUM procedure. The reason is that the set $\mathcal{U}_1$ is so small that the cost for mis-specified model is not large. 

In the second example, consider the case with larger uncertainty sets: $\mathcal{U}_1 = \{\Theta \in \mathbb{S}_+^{d}: \|\Theta\|_2 \leq 0.5 \}$. Again, we run $500$ experiments and for each run we randomly choose a $\Sigma \in \mathcal{U}_1$ and generate the post-change observations from $\mathcal{N}(0,\Sigma)$. For classic CUSUM, we randomly choose a matrix in $\mathcal{U}_1$ as the covariance matrix of its post-change normal distribution. Then, we obtain the detection delays of $T_2$ and $T_{\rm CUSUM}$. 
{\bf The mean and standard deviation of detection delay of $T_2$ is $2.06$ and $0.33$, and those of $T_{\rm CUSUM}$ is $10.28$ and $9.22$}. In this case, $T_2$ outperforms $T_{\rm CUSUM}$ since $\mathcal{U}_1$ is a large convex set and cost for a misspecified model is greater. \yc{Note that for the above two choices of $\mathcal{U}_1$, (\ref{SPPLift}) can be solved by first removing the inner maximum since the maximum is achieved at the boundary of $\mathcal{U}_1$. Then solving saddle point is equivalent to solving a convex optimization.} 



\section{Conclusions and future work}
In this paper, we propose robust detection procedures for detecting the change for mean vectors and covariance matrices, when they belong to some convex uncertainty sets. The proposed procedures are similar to classic CUSUM procedure, and the task is to determine appropriate pre-change and post-change distributions by convex optimization, which can be done efficient in both high dimensional cases. Future work includes generalizing the approach to Pollak's average worst case formulation in \cite{pollak1985optimal} and the Bayesian formulation in \cite{shiryaev1963optimum}. Ongoing work also includes generalizing the current framework to non-Gaussian distributions utilizing the results for sub-Gaussian distribution in \cite{goldenshluger2015hypothesis,guigues2016change}.

%
%
%
%
%
%
%
%
%
%
%

\bibliography{RobustChangeDetection}

\clearpage

\appendices

\section{}
In this appendix, we prove the main results. In the following, we denote $\mathbb{E}_{\xi \sim \nu}[f(\xi)]$ as the expected value of $f(\xi)$ when $\xi$ follows some distribution $\nu$. 
\begin{proof}[Proof of Theorem \ref{ARL1}]

Define that $\phi^* \triangleq -\frac{1}{2}L^*$. From Theorem 2.1 in \cite{goldenshluger2015hypothesis}, we have that 
\begin{equation}
\mathbb{E}_{\xi \sim \nu_0} [\exp(-\phi^*(\xi))] \leq \epsilon^*, ~~ \forall \nu_0 \in \mathcal{P}_0,
\label{eq1_basic}
\end{equation}
\begin{equation}
\mathbb{E}_{\xi \sim \nu_1} [\exp(\phi^*(\xi))] \leq \epsilon^*, ~~ \forall \nu_1 \in \mathcal{P}_1, 
\label{eq2_basic}
\end{equation}
where $\epsilon^*$ is the solution to the equation
\[
\mathbb{E}_{\xi \sim \nu_0^*}  [\exp(-\phi^*(\xi))]  = \mathbb{E}_{\xi \sim \nu_1^*}  [\exp(\phi^*(\xi))] , 
\] 
or equivalently, it is defined in (\ref{error}). 

Define a stopping time $T = \inf\{t>0: \sum_{i=1}^t -\phi^*(\xi_t)>b\} $, then $T_1$ in (\ref{CUSUM-like}) is the same procedure as $T$ and the arguments about $T$ are also true for $T_1$. 
Following the definition of $T$, for any $m>0$, we have that
\begin{equation}
\begin{split}
\mathbb{P}_\infty^{\nu_0} (T\leq m) \leq& \mathbb{P}_\infty^{\nu_0}  \left( \bigcup_{k=1}^m \left\{ \sum_{i=1}^k -\phi^*(\xi_i)>b \right\} \right) \\
\leq & \sum_{k=1}^m \mathbb{P}_\infty^{\nu_0}  \left( \sum_{i=1}^k -\phi^*(\xi_i)>b \right) \\
=& \sum_{k=1}^m \mathbb{P}_\infty^{\nu_0}  \left( \sum_{i=1}^k \left(-\phi^*(\xi_i)-\frac{b}{k}\right)>0\right).
\end{split}
\end{equation}

Fix $m$ and $k$, we define that $\widetilde{\phi}^* = \phi^* + b/k$ and then we use Chernoff inequality and inequality (\ref{eq1_basic}) to obtain that
\begin{equation}
\begin{split}
\mathbb{P}_{\xi \sim \nu} (-\widetilde{\phi}^*(\xi)>0) &\leq \frac{\mathbb{E}_{\xi \sim \nu}[\exp(-\widetilde{\phi}^*(\xi))]}{1} \\
&\leq \exp(-\frac{b}{k}) \epsilon^*, ~\forall \nu \in \mathcal{P}_0.
\end{split}
\end{equation}
Under $H_0$, $\xi_i \sim \nu_0 \in \mathcal{P}_0, i=1,\ldots,m$ and $\xi_i$s are independent. If we apply the shifted detector $\widetilde{\phi}^*$ on the independent variables $\xi_1, \xi_{2},\ldots, \xi_k$, from the result for $k$-repeated observations (Section 2.4 in \cite{goldenshluger2015hypothesis}) , we can have that
\[
\mathbb{P}_\infty^{\nu_0}  \left( \sum_{i=1}^k \left(-\phi^*(\xi_i)-\frac{b}{k}\right)>0\right) \leq \left( \exp\left(-\frac{b}{k}\right) \epsilon^* \right)^{k}.
\]

Then, we have that
\begin{equation}
\begin{split}
\mathbb{P}_\infty^{\nu_0} (T\leq m) \leq & \sum_{k=1}^m \left( \exp\left(-\frac{b}{k}\right) \epsilon^* \right)^{k} \\
=& \sum_{k=1}^m  \exp\left(-b\right) \left( \epsilon^* \right)^{k }, \\
=& \exp(-b) \cdot \frac{\epsilon^* - (\epsilon^*)^{m+1}}{1-\epsilon^*}.
\end{split}
\end{equation}
Letting $m$ go to infinity, we have that
\[
\mathbb{P}_\infty^{\nu_0}  (T< \infty) = \exp(-b) \cdot \frac{\epsilon^*}{(1-\epsilon^*)}.
\]

Applying Theorem $2$ in \cite{lorden1971procedures}, we have that
\[
\mathbb{E}_\infty^{\nu_0}  (T) \geq \frac{1}{\mathbb{P}_\infty^{\nu_0}  (T< \infty)} = \exp(b) \cdot \frac{1-\epsilon^*}{\epsilon^*},
\]
which concludes our result.
\end{proof}

\begin{proof}[Proof of Theorem \ref{EDD1}]

Similar with the proof for Theorem \ref{ARL1}, we define that $\phi^* = -\frac{1}{2}L^*$, $S_t = \sum_{i=1}^t -\phi^*(\xi_t)$  and a stopping time $T = \inf\{t>0: S_t>b\} $. Then $T$ is the same as $T_1$. Noticing that under $\mathbb{P}_{0}^{\nu_0,\nu_1}$, $\xi_1, \xi_2, \ldots$ is a sequence of i.i.d random variables following some distribution $\nu_1 \in \mathcal{P}_1$, the well known Wald's equality (e.g, \cite{siegmund1985sequential}) shows that
\[
\mathbb{E}_0^{\nu_0,\nu_1}[T] = \frac{\mathbb{E}_0^{\nu_0,\nu_1}[S_T]}{\mathbb{E}_{\xi_1 \sim \nu_1}[-\phi_*(\xi_1)]} = \frac{b+\mathbb{E}_0^{\nu_0,\nu_1}[S_T-b]}{\mathbb{E}_{\xi_1 \sim \nu_1}[-\phi_*(\xi_1)]},
\]
where $\mathbb{E}_0^{\nu_0,\nu_1}[S_T-b]$ is the expected overshoot above the decision boundary.

Combining (\ref{eq2_basic}) and the fact that for any $x\in \mathbb{R}$, $-x \geq 1-\exp(x)$, we have that
\[
\mathbb{E}_{\xi_1 \sim \nu_1}[-\phi_*(\xi_1)] \geq 1-\mathbb{E}_{\xi_1 \sim \nu_1} [\exp(\phi_*(\xi_1))] \geq 1-\epsilon^*.
\]
To estimate the overshoot, we apply (8.18) and (8.50) in \cite{siegmund1985sequential} to show that as $b \rightarrow \infty$, the following limit holds,
\[
\mathbb{E}_0^{\nu_0,\nu_1}[S_T-b] \rightarrow \frac{\mathbb{E}_{\xi_1\sim \nu_1}[\phi^*(\xi_1)^2]}{2\mathbb{E}_{\xi_1 \sim \nu_1}[\phi^*(\xi_1)]} - \sum_{n=1}^{\infty} \frac{\mathbb{E}_0^{\nu_0,\nu_1}[S_n^-]}{n},
\]
where $x^- = -\min(x,0)$.

By the assumption made in the statement, we have that for some $M>0$,
$
\mathbb{E}_{\xi_1 \sim \nu_1}[\phi_*^2(\xi_1)] \leq M.
$
Therefore, as $b\rightarrow \infty$, we have that $\mathbb{E}_0^{\nu_0,\nu_1}[S_T-b] = o(b)$. Combing the Theorem 2 in \cite{lorden1971procedures}, we conclude the result.

\end{proof}

\begin{proof}[Proof of Corollary \ref{ARL2} and \ref{EDD2}]
When $\phi^*$ is obtained from (\ref{iquaddet}), from the Proposition 4.1 in \cite{guigues2016change}, we have that 
\begin{equation}
\mathbb{E}_{\xi \sim \nu_0} [\exp(-\phi^*(\xi))] \leq \epsilon^*, ~~ \forall \nu_0 \in \mathcal{G}_0,
\label{eq1_quadratic}
\end{equation}
\begin{equation}
\mathbb{E}_{\xi \sim \nu_1} [\exp(\phi^*(\xi))] \leq \epsilon^*, ~~ \forall \nu_1 \in \mathcal{G}_1, 
\label{eq2_quadratic}
\end{equation}
where $\epsilon^*$ is defined in (\ref{error-2}). 
Then, following the same proof routine as Theorem \ref{ARL1} and \ref{EDD1}, we conclude the results.

\end{proof}

\end{document}